\documentclass[12pt,a4paper]{article}
\usepackage[utf8]{inputenc}
\usepackage{graphicx}
\usepackage{amsmath}
\usepackage{amssymb}
\usepackage{amscd}
\usepackage{amsthm}
\usepackage{cleveref}

\newtheorem{theorem}{Theorem}[section]
\newtheorem{corollary}[theorem]{Corollary}
\newtheorem{lemma}[theorem]{Lemma}
\newtheorem{proposition}[theorem]{Proposition}

\newtheorem{example}[theorem]{Example}
\newtheorem{definition}[theorem]{Definition}
\newtheorem{remark}[theorem]{Remark}
\newcommand{\EQ}{\begin{equation}}
\newcommand{\EN}{\end{equation}}

\newcommand{\zero}{{\mathbf{0}}}
\newcommand{\one}{{\mathbf{1}}}
\newcommand{\two}{{\mathbf{2}}}
\newcommand{\by}{{\bf y}}
\newcommand{\bx}{{\bf x}}
\newcommand{\bv}{{\bf v}}
\newcommand{\bu}{{\bf u}}
\newcommand{\bw}{{\bf w}}
\newcommand{\Z}{\mathbb{Z}}
\newcommand{\dd}{\displaystyle}
\newcommand{\add}{\Z_2\Z_4}
\newcommand{\C}{{\cal C}}

\title{On $\Z_2\Z_4$-additive complementary dual codes and related LCD codes}

\author{N. Benbelkacem\thanks{Department of Algebra and Number Theory, University of Science and Technology Houari Boumediene, Algeria.}, J. Borges\thanks{Department of Information and Communications Engineering, Universitat Aut\`{o}noma de Barcelona, 08193-Bellaterra, Spain.}, S.T. Dougherty\thanks{Department of Mathematics, University of Scranton, USA.} and C. Fern\'{a}ndez-C\'{o}rdoba\footnotemark[2]}

\begin{document}
	\maketitle
	
\begin{abstract}
Linear complementary dual codes were defined by Massey in 1992, and were used to give an optimum linear coding solution for the two user binary adder channel.  In this paper, we define the analog of LCD codes over fields in the ambient space with mixed binary and quaternary alphabets.  These codes are additive, in the sense that they are additive subgroups, rather than linear as they are not vector spaces over some finite field.  We study the structure of these codes and we use the canonical Gray map from this space to the Hamming space to construct binary LCD codes in certain cases.
		
\end{abstract}

\section{Introduction}
Let $\Z_2$ and $\Z_4$ be the ring of integers modulo 2 and 4 respectively. Let $\Z_2^n$ denote the set of all binary vectors of length $n$ and let $\Z_4^n$ be the set of all $n$-tuples over the ring $\Z_4$. Any non-empty subset $C$ of $\Z_2^n$ is a binary code and a subgroup of $\Z_2^n$ is called a {\it binary linear code}. If $C\subset \Z_2^n$ is a binary linear code, then $C$ is also a linear subspace of $\Z_2^n$ (considered as the binary vector space of dimension $n$). If the dimension of $C$ is $k$, then we say that $C$ is a binary $(n,k)$ code. Any non-empty subset ${\cal C}$ of $\Z_4^n$ is a quaternary code and a subgroup of $\Z_4^n$ is	called a {\it quaternary linear code}.
	
A binary code is said to be {\em linear complementary dual} (LCD) if it is linear and $C \cap C^\perp =\{\zero\}$. Binary LCD codes were defined and characterized in \cite{Massey}. In that paper, it is shown that these codes are an optimum linear coding solution for the two-user binary adder channel. In \cite{Carlet}, LCD codes are used for an application in the security of digital communication; specifically, in counter measures to passive and active side-channel attacks in embedded cryptosystems.

Some constructions of LCD codes are given in \cite{Dougherty}, where the authors develop a linear programming bound on the largest size of an LCD code when the length and the minimum distance is given. LCD cyclic codes over finite fields are called reversible codes in \cite{Massey-2}, and it is shown that some LCD cyclic codes over finite fields are BCH codes. Later, in \cite{LiDingLi}, the authors give constructions of families of reversible cyclic codes over finite fields and prove that some of these LCD cyclic codes are optimal codes. A more general class of LCD cyclic codes, the LCD quasi-cyclic codes, were studied in \cite{Morteza} and generalized in \cite{Quasi}. Complementary dual codes have also been studied in \cite{Chain} for linear codes over finite chain rings. 

A $\add$-additive code is an additive subgroup of $\Z_2^\alpha\times\Z_4^\beta$.
These codes were first introduced in \cite{transProper}, as abelian translation-invariant propelinear codes. Later, an exhaustive description of $\add$-additive codes
was done in \cite{AddDual}. The structure and properties
of $\add$-additive codes have been intensely studied, for
example in \cite{MDS}, \cite{z2z4SD}, and \cite{Z2Z4RK}.

Here, we generalize the notion of LDC codes to {\em additive complementary dual} (ACD) codes in $\Z_2^\alpha\times\Z_4^\beta$. We construct infinite families of codes that are ACD. We use these ACD codes to construct infinite families of binary LCD codes via the Gray map. We give conditions for the case when the image of an ACD code is a binary LCD code.
	
The paper is organized as follows. In Section \ref{additive}, we recall the basic definitions and properties of $\add$-additive codes. In Section \ref{ACD}, we define ACD codes and state several properties similar to those of LCD codes. In Section \ref{cases}, we study different cases of ACD codes, taking into account if the binary and the quaternary parts are complementary dual codes, or are not. We give several infinite families of ACD codes. In Section \ref{section:image}, we construct binary LCD codes as binary images of ACD codes, giving necessary and sufficient conditions. Finally, in \Cref{conclusions}, we summarize the main results of the paper.

\section{$\add$-additive codes}\label{additive}
We define $\add$-additive codes and give some known properties. Specially, those properties related to the duality and linearity of $\add$-codes and also to the binary image of these codes via the Gray map.

	\begin{definition}
		A $\add$-{\em additive code} $\C$ is a subgroup of $\Z_2^\alpha\times\Z_4^\beta$, where $\alpha$, and $\beta$ are nonnegative integers.
	\end{definition}
	
	 Unless otherwise stated, the first $\alpha$ coordinates correspond to the coordinates over $\Z_2$ and the last $\beta$ coordinates to the coordinates over $\Z_4$. For a vector $\bu \in \Z_2^\alpha\times\Z_4^\beta$, we write $\bu=(u\mid u')$, where $u=(u_1,\dots,u_\alpha)\in\Z_2^\alpha$ and $u'=(u'_1,\dots,u'_\beta)\in\Z_4^\beta$. Note that a vector in boldface indicates that it has mixed coordinates (binary and quaternary), whereas if a vector is not in boldface, then all of its coordinates are either over $\Z_2$ or over $\Z_4$. However, when we write $\zero$ (respectively $\one$), we mean the binary or quaternary all-zero vector (respectively all-one vector). In such cases, the alphabet and the length of these vectors will be clear from the context. Also, we denote by $\two$ the quaternary all-two vector. 
	
	If $\C$ is a $\add$-additive code, then $\C$ is isomorphic to
	$\Z_2^{\gamma}\times \Z_4^{\delta}$, for some nonnegative integers $\gamma$ and $\delta$. Therefore, $\C$ has
	$|\C|=2^{\gamma+2\delta}$ codewords and the number of codewords of order at most two
	is $2^{\gamma+\delta}$.
	
	Let $X$ (respectively $Y$) be the set of $\Z_2$ (respectively
	$\Z_4$) coordinate positions, so $|X|=\alpha$ and $|Y|=\beta$. Call
	$\C_X$ (respectively $\C_Y$) the punctured code of $\C$ by deleting the coordinates outside $X$
	(respectively $Y$).  Let $\C_b$ be the subcode of
	$\C$ which contains exactly all codewords of order at most $2$. Let $\kappa$ be
	the dimension of $(\C_b)_X$, which is a binary linear code. For the case
	$\alpha=0$, we write $\kappa=0$.
	
	Considering all these parameters, we say that $\C$ is of {\em type}
	$(\alpha,\beta;\gamma,\delta;\kappa)$. Notice that $\C_Y$ is a quaternary linear code of type
	$(0,\beta;\gamma_Y,\delta;0)$, where $0\leq \gamma_Y \leq \gamma$, and $\C_X$ is a binary linear code of type
	$(\alpha,0;\gamma_X,0;\gamma_X)$, where $\kappa \leq \gamma_X \leq \kappa+\delta$. A  $\add$-additive code $\C$ is said to be {\em separable} if $\C=\C_X\times\C_Y$.
	
	Two $\add$-additive codes $\C_1$ and $\C_2$, both of type
	$(\alpha,\beta;\gamma,\delta;\kappa)$, are said to be {\it
		monomially equivalent}, if one can be obtained from the other by permutating
	the coordinates and (if necessary) changing the signs of certain
	$\Z_4$ coordinates.
	
	Although $\C$ is not a free module in general, there exist $\{\mathbf{u}_i\}_{i=1}^\gamma$ and $\{\mathbf{v}_j\}_{j=1}^\delta$ such that every codeword
 in $\C$ can be uniquely expressible in the form
$$\dd \sum_{i=1}^{\gamma}\lambda_i\bu_i+\sum_{j=1}^{\delta}\mu_j\bv_j,$$
where $\lambda_i \in \Z_2$ for all $1\leq i\leq \gamma$,
$\mu_j\in\Z_4$ for all $1\leq j\leq \delta$ and $\bu_i,
\bv_j$ are codewords of order two and order four, respectively. Moreover, the vectors $\bu_i, \bv_j$ give us a
	generator matrix $G$ of size $(\gamma+\delta)\times (\alpha+\beta)$ for
	the code $\C$. This generator matrix  $G$  can be written as \EQ \label{eq:matrixG}
	G= \left(\begin{array}{c|c} G_X& G_Y\end{array}\right ),
	\EN %
	where $G_X$ is matrix over $\Z_2$ of size $(\gamma+\delta)\times \alpha$ and $G_Y$ is 
	a matrix over $\Z_4$ of size $\gamma+\delta\times \beta$. Note that $G_X$ is the generator matrix of $\C_X$ and $G_Y$ is the generator matrix of $\C_Y$.
	
	In \cite{AddDual}, it is proven that if $\C$ is a $\add$-additive code of type
	$(\alpha,\beta;\gamma,\delta;\kappa)$, then $\C$ is
	permutation equivalent to a $\add$-additive code with standard generator matrix
	of the form
	\EQ \label{eq:StandardForm}
	G_S= \left ( \begin{array}{cc|ccc}
		I_{\kappa} & T_b & 2T_2 & \zero & \zero\\
		\zero & \zero & 2T_1 & 2I_{\gamma-\kappa} & \zero\\
		\hline \zero & S_b & S_q & R & I_{\delta} \end{array} \right ), \EN \noindent where $T_b, S_b$ are matrices over
	$\Z_2$;  $T_1, T_2, R$ are matrices over $\Z_4$ with all entries in $\{0,1\}\subset \Z_4$; and $S_q$ is a matrix
	over $\Z_4$.

Let $\phi:\Z_4 \longrightarrow \Z_2^{2}$ be the usual Gray map, $\phi$ defined in \cite{sole}, that is, $$\phi(0)=(0,0),\ \phi(1)=(0,1), \ \phi(2)=(1,1), \
\phi(3)=(1,0).$$
Let $\Phi$ be the following extension of the usual Gray map: $\Phi:
\Z_2^{\alpha}\times\Z_4^{\beta} \longrightarrow \Z_2^{n}$ given by
$$
\Phi(u\mid u')=(u,\phi(u'_1),\ldots,\phi(u'_\beta))\;\;\;\forall
u\in\Z_2^\alpha,\;\forall u'=(u'_1,\ldots,u'_\beta)\in \Z_4^\beta;
$$
where $n=\alpha+2\beta$. Let $\C$ be a $\add$-additive code of type $(\alpha,\beta;\gamma,\delta;\kappa)$. The binary code $C=\Phi(\C)$ of length $n=\alpha+2\beta$ is a {\it $\add$-linear} code of type $(\alpha,\beta;\gamma,\delta;\kappa)$. From now on, we denote the code in $\Z_2^\alpha\times\Z_4^\beta$ by using a calligraphic letter, $\C$, whereas the image under the Gray map is written in a regular letter $C=\Phi(\C)$.

	The inner
	product of two vectors $\bu=(u\mid u')$ and $\bv=(v\mid v')$  is defined as
	$$ [\bu,\bv] =2(\sum_{i=1}^{\alpha} u_iv_i)+\sum_{j=1}^{\beta}
	u'_jv'_j \ \in \ \Z_4,$$ where the computations are made taking the zeros and ones in the $\alpha$ binary coordinates as quaternary zeros and ones, respectively. For binary vectors $u,v$ and quaternary vectors $u',v'$, we denote its inner product as $[u,v]_2$ and $[u',v']_4$, respectively. Note that $[\bu,\bv] =2[u,v]_2+[u',v']_4$, for $\bu=(u\mid u')$, $\bv=(v\mid v')$, and considering $[u,v]_2\in\{0,1\}$ as elements in $\Z_4$. The {\it dual code}	of $\C$, denoted by ${\cal C}^\perp$, is defined in the standard
	way as
	$$
	{\cal C}^\perp=\{\bv\in \Z_2^\alpha \times \Z_4^\beta \;|\;
	[\bu,\bv]  =0 \mbox{ for all } \bu\in {\cal C}\}.
	$$
	
Let $\C$ be a $\add$-additive code. Then, its dual code is also a $\add$-additive code. Moreover, their types are related as it is shown in the following theorem.

	\begin{theorem}[$\cite{AddDual}$]\label{parameters} If $\C$ is a $\add$-additive code of type $(\alpha,\beta;\gamma,\delta;\kappa)$, then its dual code, $\C^\perp$, is a $\add$-additive code of type $(\alpha,\beta;\bar{\gamma},\bar{\delta};\bar{\kappa})$,
		where $$\begin{array}{l} \bar{\gamma} = \alpha + \gamma - 2\kappa,\\
		\bar{\delta} =\beta - \gamma - \delta + \kappa,\\
		\bar{\kappa}=\alpha-\kappa. \end{array}$$
	Moreover, if ${\cal C}$ is a $\add$-additive separable code, then ${\cal C}^\perp=({\cal C}_X)^\perp\times ({\cal
	C}_Y)^\perp$.
	\end{theorem}
	
As for binary and quaternary codes, we say that a $\add$-additive code $\C$ is {\em self-orthogonal} if $\C\subset\C^\perp$ and {\em self-dual} if $\C=\C^\perp$.	
	
In general, the Gray map image $C$ of $\C$ is not linear, and thus, it need not be related to a dual in the usual sense. We then define the {\it $\add$-dual} of $C$ to be the code $C_{\perp}=\Phi(\C^\perp)$. We have the following diagram
\[
\begin{CD}
\C @>\Phi>> C=\Phi(\C) \\
@V {\perp}  VV \\
\C^\perp @>\Phi >> C_{\perp}=\Phi(\C^\perp)
\end{CD}
\]
\noindent Note that, in general, $C_\perp$ is not the dual of $C$, so we cannot add an arrow on the right side to produce a commutative diagram. For example, if we have a $\add$-linear code $C$ that is not linear, then clearly $C^\perp\not=C_\perp$. 

The linearity of $\add$-linear codes was studied in \cite{Z2Z4RK}. The key to establish this linearity was the fact that
\begin{equation} \label{equation:linear}
\Phi(\bv+\bw) =
\Phi(\bv) + \Phi(\bw) + \Phi(2\bv*\bw),
\end{equation}
where $*$ denotes the component-wise product. It follows immediately that $\Phi(\C)$ is
linear if and only if $2 {\bf v} * \bw \in \C$, for all ${\bf v}, \bw\in\C$. Note that $\Phi(\bv)+\Phi(2\bu)=\Phi(\bv+2\bu)$, for $\bv,\bu\in\Z_2^\alpha\times\Z_4^\beta$.

	Let $\C$ be a $\add$-additive code with generator matrix as in (\ref{eq:matrixG}). We define the product 
	$$G\cdot G^T= \left(\begin{array}{c|c}G_X&G_Y\end{array}\right)\cdot \left(\begin{array}{c}G^T_X\\ \hline G^T_Y\end{array}\right)=2G_XG^T_X+G_YG^T_Y,
	$$
	with entries from $\Z_4$, where all entries in $G_X$ are considered as elements in $\{0,1\}\subset\Z_4$ and the product of a row by a column is computed as the inner product of vectors in $\Z_2^\alpha\times \Z_4^\beta$. Note that in $G_XG^T_X$ and $G_YG^T_Y$ the usual matrix multiplication is used, but not in $G\cdot G^T$.

\section{Additive complementary dual codes}\label{ACD}

In this section, we generalize the concept of linear complementary duality to $\add$-additive codes. We also give an infinite family of $\add$-additive codes that are additive complementary dual.

\begin{definition}
A code $\C\subseteq \Z_2^\alpha\times\Z_4^\beta$ is {\em additive complementary dual} (briefly ACD) if it is a $\add$-additive code such that $\C\cap \C^\perp=\{\zero\}$.
\end{definition}

For the case $\beta=0$, an ACD code is simply a binary LCD code. If $\alpha=0$, then an ACD code is called a quaternary LCD code.
	
For binary LCD codes with have the following property.

\begin{lemma}[\cite{Massey}]
Let $C$ be a binary LCD code. Then $\Z_2^n=C\oplus C^\perp.$ That is, any vector $w$ in $\Z_2^n$ can be written uniquely as $w_1+w_2$, for $w_1\in C$ and $w_2\in C^\perp$.
\end{lemma}

This proposition can be easily extended to ACD codes.

\begin{lemma}\label{sumadirecta}
Let $\C\subseteq \Z_2^\alpha\times\Z_4^\beta$ be an ACD code. Then any vector $\bw\in\Z_2^\alpha\times\Z_4^\beta$ can be written uniquely as $\bw_1+\bw_2$, for $\bw_1\in\C$ and $\bw_2\in\C^\perp$.
\end{lemma}

The following proposition characterizes binary LCD codes by considering their generator and parity-check matrices.

\begin{proposition}[\cite{Massey}]\label{Massey-matrix}
Let $C$ be a binary $(n,k)$ linear code with generator matrix $G$ and parity-check matrix $H$. The following statements are equivalent:
\begin{enumerate}
    \item $C$ is an LCD code,
    \item the $k\times k$ matrix $GG^T$ is nonsingular,
    \item the $(n-k)\times (n-k)$ matrix $HH^T$ is nonsingular.
\end{enumerate} 
\end{proposition}

Now, we will study an analogous property of \Cref{Massey-matrix} for ACD codes.

If we have a generator matrix for an ACD code, sometimes we will need to make some changes in the matrix (e.g. permutations of rows and columns and sign changes). Hence, we will make use of the following technical lemma.
	
	\begin{lemma}
		If $\C$ is an ACD code, then any monomially equivalent code $\C'$ is also an ACD code.
	\end{lemma}
	
	\begin{proof}
		The result is straightforward, since the inner product of two vectors is not changed by coordinate permutations or sign changes.
	\end{proof}
	
	Now, we state sufficient conditions for a $\add$-additive code to be ACD.
	
	\begin{proposition}\label{rows}
		Let $G$ be a generator matrix for a $\add$-additive code $\C$. Denote by $\bv_1,\ldots,\bv_r$ the rows of $G$, so that $\C=\langle \bv_1,\ldots,\bv_r \rangle$. If $[\bv_i,\bv_j]\in\{0,2\}$ and $[\bv_i,\bv_i]\notin\{0,2\}$ for all $i,j=1,\ldots,r$ such that $i\neq j$, then $\C$ is an ACD code and $\C_Y$ is a quaternary LCD code.
	\end{proposition}
	
	\begin{proof}
		Let $\bx\in\C\setminus\{\zero\}$ be any nonzero codeword. We want to show that $\bx\notin\C^\perp$. Since $\bx\in\C$, $\bx$ can be written as $\bx=\sum_{i\in J} \lambda_i \bv_i$, where $J=\{1,\ldots,r\}$ and $\lambda_i\in\Z_4$.
		
		First, assume there exists $j\in J$ such that $\lambda_j\in\{1,3\}$. Thus,
		$$
		[\bx,\bv_j]=\sum_{i\in J} \lambda_i[\bv_i,\bv_j]=\sum_{i\in J \setminus\{j\}}
		\mu_i +\lambda_j[\bv_j,\bv_j].
		$$
		Since $\mu_i=\lambda_i[\bv_i,\bv_j]\in\{0,2\}$, we have that $[\bx,\bv_j]\not=0$ and $\bx\notin\C^\perp$.
		
		Finally, if $\lambda_i\in\{0,2\}$ for all $i\in J$, let $j\in J$ such that $\lambda_j=2$. Then, $[\bx,\bv_j]=\sum_{i\in J} \lambda_i[\bv_i,\bv_j]=2[\bv_j,\bv_j]=2$ and $\bx\notin\C^\perp$.
		
		The same argument can be used for $\C_Y$.
	\end{proof}
	
	\begin{corollary}\label{matrix}
		Let $G$ be a generator matrix for a $\add$-additive code $\C$ and consider the matrix $G\cdot G^T=(w_{ij})_{1\leq i,j \leq r}$ with entries from $\Z_4$. If $w_{ij}\in\{0,2\}$ and $w_{ii}\notin\{0,2\}$ for all $i,j=1,\ldots,r$ such that $i\neq j$, then $\C$ is an ACD code and $\C_Y$ is a quaternary LCD code.
	\end{corollary}
	
	\begin{proof}
		It is a direct consequence of \Cref{rows}.
	\end{proof}
	
	\begin{remark}\label{remark:NotProp}
		The reverse statements of \Cref{rows} and \Cref{matrix} are not true in general. Let $\C$ be the $\add$-additive code generated by
		$$
		G=\left(\begin{array}{ccc|cc}
		1&1&1&2 & 0\\
		0&0&1&2 &1\\
		
		\end{array}
		\right).
		$$
		We have that $\C$ is ACD, but in this case
		$$G\cdot G^T=\left(\begin{array}{cc}
		2 & 2\\
		2 &0\\
		\end{array}
		\right).
		$$
	\end{remark}
	
	\begin{corollary}\label{invertible}
		If $G\cdot G^T$ is invertible (over $\Z_4$), then $\C$ is an ACD code.
	\end{corollary}

	\begin{remark}
		Again, the reverse statement of \Cref{invertible} is not true in general. Let $\C$ be the $\add$-additive code generated by
		$$
		G=\left(\begin{array}{ccc|cc}
		1&0&0&0 & 0\\
		0&1&1&2 &1\\
		0&0&0&0 & 2\\
		\end{array}
		\right).
		$$ 
		We have that $\C$ is ACD, but in this case
		$$G\cdot G^T=\left(\begin{array}{cc}
		3 & 2\\
		2 &0\\
		\end{array}
		\right),
		$$
		that is not invertible (over $\Z_4$).
	\end{remark}
	
	\begin{proposition}\label{free-ACD}
	Let $C$ be a binary $(\alpha,k)$ code and let $\{v_1,\dots,v_k\}$ be a basis for $C$. Let $\delta\geq k$ and let $G_X$ be the $\delta\times\alpha$ matrix whose non-zero row vectors are $\{v_1,\dots,v_k\}$. Then, the $\add$-additive code $\C$ of type $(\alpha,\delta;0,\delta;0)$ generated by
	$$
	G=\left(G_X\mid I_\delta\right)
	$$
	is an ACD code.
	\end{proposition}
\begin{proof}
	Let $\C$ be the $\add$-additive code with generator matrix $G=\left(G_X\mid I_\delta\right)$. Note that $G\cdot G^T=2G_XG_X^T+I_\delta=(w_{ij})_{1\leq i,j \leq \delta}$, where all entries in $G_X$ are considered as elements in
	$\{0, 1\} \subset \Z_4$, and hence $2G_XG_X^T$ has all entries in $\{0,2\}$. Therefore, $w_{ij}\in\{0,2\}$ and $w_{ii}\notin\{0,2\}$ for all $i,j=1,\ldots,\delta$ such that $i\neq j$, and $\C$ is ACD by \Cref{matrix}. The generator matrix $G$ is in standard form and it is easy to see that $\C$ is of type $(\alpha,\delta;0,\delta;0)$.
\end{proof}
	
\section{Complementary duality of $\C$, $\C_X$ and $\C_Y$}\label{cases}

In this section, we discuss the complementary duality of a $\add$-additive code $\C$ in terms of the complementary duality of $\C_X$ and $\C_Y$. First, in \Cref{ex:NonACD}, we give an example of a non ACD code $\C$ such that $\C_X$ and $\C_Y$ are binary and quaternary LCD codes, respectively. We also show in \Cref{ex:G-H} an ACD code, where neither $\C_X$ nor $\C_Y$ are LCD codes. The case when $\C$ is an ACD code and both $\C_X$ and $\C_Y$ are LCD codes is studied in \Cref{sec:separable}, and it includes separable codes. Finally, in \Cref{sec:OnlyOne}, we consider the case when $\C$ is ACD and either $\C_X$ or $\C_Y$ is LCD.
	
	We start with an example of a $\add$-additive code that is not ACD and, by contrast, $\C_X$ is a binary LCD code and $\C_Y$ is a quaternary LCD code.
		\begin{example}\label{ex:NonACD}
		Let $\mathcal{C}$ be a $\add$-additive code of type $(\alpha,\alpha;1,\alpha;\alpha)$ generated by
		$$
		\left(\begin{array}{c|c}
		I_\alpha&I_\alpha\\
		\bf{1}&\bf{2}
		\end{array}
		\right).
		$$
		Clearly, $\C_X=\Z_2^\alpha$ is an LCD code and $\C_Y=\Z_4^\alpha$ is also LCD. However, the last row $\bv_{\alpha+1}=(\bf{1}\mid \bf{2})$ is orthogonal to any row in the generator matrix. Hence, $\bv_{\alpha+1}\in \C\cap \C^\perp$ and $\C$ is not complementary dual.
	\end{example}
	
	In the following example we show that there exist $\add$-additive ACD codes $\C$ such that neither $\C_X$ nor $\C_Y$ are complementary dual codes.
	\begin{example}\label{ex:G-H}
		Let $\mathcal{C}$ be the code given in $\Cref{remark:NotProp}$. The generator and the parity check matrices of $C$ are
		$$
		G=\left(\begin{array}{ccc|cc}
		1&1&1&2 & 0\\
		0&0&1&2 &1\\
		\end{array}
		\right), \hspace{1truecm}
		H=\left(\begin{array}{ccc|cc}
		1&0&1&0&2\\
		0&1&1&0&2\\
		0&0&1&1&0\\
		\end{array}
		\right),
		$$
		respectively. Note that $(1,1,0)\in\C_X\cap\C_X^\perp$ and $(2,0)\in\C_Y\cap\C_Y^\perp$. Therefore, $\C_X$ and $\C_Y$ are not a binary LCD and quaternary LCD codes, respectively. However, we have that $\C$ is an ACD code since $\C\cap\C^\perp =\{\zero\}$.
	\end{example}

\subsection{ACD codes $\C$ with both $\C_X$ and $\C_Y$ LCD codes}\label{sec:separable}

For this case, we distinguish when $\C$ is a separable or a non-separable code.

\begin{proposition}\label{prop:separable}
Let $\C$ be a $\add$-additive code. If $\C$ is separable, then $\C$ is an ACD code if and only if $\C_{X}$ is a binary LCD code and  $\C_{Y}$ is a quaternary LCD code.
\end{proposition}
	
\begin{proof}
	Since $\C$ is separable, $\C= \C_{X} \times \C_{Y}$ and $\C^{\perp }= \C_{X}^{\perp } \times \C_{Y}^{\perp }$.
	Assume $\C$ is an ACD code. By definition, any codeword $\bu=(u\mid u') \in \C \cap \C^{\perp} $ is the zero codeword. Let $u \in \C_{X} \cap \C_{X} ^{\perp}$. The codeword $ (u\mid \zero)\in \C \cap \C^{\perp }$  and this implies that $u=\zero$ . So, $\C_{X}$  is a binary LCD code. Similarly, if $u' \in \C_{Y} \cap \C_{Y} ^{\perp}$, then $ (\zero\mid u')\in \C \cap \C^{\perp}$ and hence $u'=\zero$ that implies that $\C_{Y}$ is a quaternary LCD code.
		
	Conversely, let $\C_{X}$ and $\C_{Y}$ be a binary and a quaternary LCD code, respectively. Let $\bu=(u\mid u') \in \C \cap \C^{\perp }$. This implies that $u\in \C_{X} \cap \C_{X} ^{\perp }=\{\zero\}$ and $u'\in \C_{Y} \cap \C_{Y}^{\perp }\{\zero\}$. Then, $\bu=(u\mid u')$ is the zero codeword and $\C$ is an ACD code.
\end{proof}

	As we have seen in \Cref{prop:separable}, a separable $\add$-additive code $\C=\C_X \times \C_Y$ is ACD if and only if both $\C_X$ and $\C_Y$ are linear complementary codes. However, there also exist   non-separable ACD codes $\C$ such that $\C_X$ is binary LCD and $\C_Y$ is quaternary LCD, as we can see in the following example.

	\begin{example}\label{first}
		Let $\mathcal{C}$ be a $\add$-additive code generated by
		$$
		\left(\begin{array}{ccc|ccc}
		1&0&0&1&2&0\\
		0&1&0&0&2&1\\
		0&0&1&2&1&2
		\end{array}
		\right).
		$$
		Let $\bv_1,\bv_2$ and $\bv_3$ the row vectors of $G$. We can see that $\bv_i\bv_j=0$, for all $i\not=j$, and $\bv_i\bv_i\not\in\{0,2\}$, Therefore, by \Cref{rows}, $\C$ is ACD. Moreover, $\C_X$ and $\C_Y$ are both LCD codes.
	\end{example}

	\subsection{ACD codes $\C$ with either $\C_X$ or $\C_Y$ LCD codes}\label{sec:OnlyOne}

	In \Cref{free-ACD}, we obtain a family of free $\add$-additive codes in $\Z_2^\alpha\times\Z_4^\beta$ that are ACD. The following lemma is a particular case of \Cref{free-ACD} and gives a family of ACD codes $\C$ where $\C_X$ is self-orthogonal and $\C_Y=\Z_4^\delta$, Therefore, $\C_X$ is not LCD and $\C_Y$ is LCD.
	
	\begin{lemma}\label{s-o}
		Let $D$ be an $(\alpha,\delta)$ binary self-orthogonal code with generator matrix $G_X$. Then, the $\add$-additive code $\C$ with generator matrix
		$$
		G=\left(G_X\mid I_\delta\right)
		$$
		is an ACD code of type $(\alpha,\delta;0,\delta;0)$.
	\end{lemma}
	
	\begin{corollary}
		For any non-negative integers $\alpha$ and $\delta\leq\lfloor \frac{\alpha}{2}\rfloor$, there exist an ACD code of type $(\alpha,\delta;0,\delta;0)$ and $(\alpha,\delta; \alpha-2\delta,\delta;\alpha-\delta)$.
	\end{corollary}
	\begin{proof}
		There exist an $(\alpha,\delta)$ binary self-orthogonal code $D$, for all $\delta\leq\lfloor \frac{\alpha}{2}\rfloor$. By using the construction given in \Cref{s-o}, we obtain an ACD code $\C$ of type $(\alpha,\delta;0,\delta;0)$. The code $\C^\perp$ is also ACD and it is of type $(\alpha,\delta; \alpha-2\delta,\delta;\alpha-\delta)$ by \Cref{parameters}.
	\end{proof}
	
	\begin{corollary}
		Let $C$ be a $(\alpha,\delta)$ binary self-dual code with generator matrix $G_X$. Then, the $\add$-additive dual code $\C$ with generator matrix
		$$
		G=\left(G_X\mid I_{\frac{\alpha}{2}}\right)
		$$
		is an ACD code of type $(\alpha,\frac{\alpha}{2};0,\delta;0)$.
	\end{corollary}
	
	Note that if $\C_Y$ is a quaternary self-orthogonal (or self-dual in particular) code with generator matrix $G_Y$, then the $\add$-additive code with generator matrix $G=(I_\alpha\mid G_Y)$ is not necessarily ACD.
	
	\begin{example}\label{ex:xeq0}
		Let $\mathcal{C}$ be a $\add$-additive code generated by
		$$
		\left(\begin{array}{ccc|cccc}
		1&0&0&1 & 1 & 1 & 1\\
		0&1&0&2&0&2&0\\
		0&0&1&0&2&0&2
		\end{array}
		\right)
		$$
		Note that $\C_Y$ is a quaternary self-dual code, but $\C$ is not ACD since $(0,0,0,2,2,2,2)\in\C\cap\C^\perp$.
	\end{example}
	
	In general, if $\C_Y$ is a quaternary self-orthogonal code of length $\beta$ and type $2^\gamma4^\delta$ with generator matrix $G_Y$ and
	$$
	G=\left(I_{\gamma+\delta}\mid G_Y\right)
	$$
	is the generator matrix of a $\add$-additive code $\C$, then we have that $G\cdot G^T=2I_{\gamma+\delta}$, and hence we do not know in general whether $\C$ is an ACD code or not. The last example gives a code that is not ACD; however, there are some cases where the code $\C$ is an ACD code as in the following proposition that is easily proven.
	
	\begin{proposition}\label{ex:xnot0}
		Let $\C$ be the $\add$-additive complementary dual code of type $(\alpha,\alpha,\alpha,0;\alpha)$ generated by
		$$
		G=\left(I_\alpha\mid 2I_\alpha\right).
		$$
		Then, $\C_X$ is a binary LCD code and $\C_Y$ is not a quaternary LCD code because it is a self-dual code.
	\end{proposition}
	
	\begin{lemma}\label{lemma:xnot0}
		Let $\C$ be a $\add$-additive code such that $\C_X$ is a binary LCD code and $\C_Y$ is a quaternary self-orthogonal code. Then $\C$ is an ACD code if and only if for all $\bw=(w\mid w')\in \C$, if $w'\not=\zero$, then $w\not=\zero$.
	\end{lemma}	
	
	\begin{corollary}
		Let $\C$ be a $\add$-additive code of type $(\alpha,\beta;\gamma,\delta;\kappa)$ such that $\C_X$ is a binary LCD code and $\C_Y$ is a quaternary self-orthogonal code. If $\delta\not=0$, then $\C$ is not an ACD code.
	\end{corollary}
	\begin{proof}
		If $\delta\not=0$, then there exist $\bw=(w\mid w')\in \C$ of order $4$. Then $2\bw=(0\mid 2w')\in \C$ and, by \Cref{lemma:xnot0}, $\C$ is not an ACD code.
	\end{proof}
	
	In \Cref{ex:xeq0}, $\C_X$ is an LCD code, $\C_Y$ is self-orthogonal and the value $\delta$ in the $\add$-additive code $\C$ is not $0$. Therefore, the code is not an ACD code.
	
	\begin{corollary}
		Let $\C$ be a $\add$-additive code of type $(\alpha,\beta;\gamma,\delta;\kappa)$ such that $\C_X$ is a binary LCD code and $\C_Y$ is a quaternary self-orthogonal code. If $\delta=0$ and $\kappa=\gamma$, then $\C$ is an ACD code.
	\end{corollary}
	
	In \Cref{ex:xnot0}, we have an ACD code of type $(\alpha,\beta;\gamma,\delta;\kappa)=(\alpha,\alpha,\alpha,0;\alpha)$. Note that $\C_X$ is binary LCD, $\C_Y$ is quaternary self-dual, $\delta=0$ and $\kappa=\gamma$.

\section{Binary LCD codes from ACD codes}\label{section:image}

 In general, if $\C$ is a $\add$-additive code, $C=\phi(\C)$ may not be a linear code. Therefore, if $\C$ is an ACD code, then $C$ is not necessarily a binary LCD code. In this section, we establish the conditions for an ACD code $\C$ so that its image $C$ is a binary LCD code.\\
 
 We first give three examples of ACD codes with different situations when considering their binary image. In the first case, the binary images under the Gray map of both, the code and its dual, are not linear. In the second and third case, the binary image of the code is linear whereas the binary image of its dual is not linear. In the second example the binary image is LCD but it is not LCD in the third example. At the end of the section we will see that if the binary image under the Gray map of the code and its dual are linear, then the code is necessarily LCD.
 
 \begin{example}\label{ex:ACD-nonlinears}
 Let $\mathcal{C}_1$ be the $\add$-additive code with generator and parity check matrix 
		$$
		\left(\begin{array}{cc|cccc}
		0&1 & 2 & 3 & 1 & 0\\
		1&1 & 1 & 3 & 0 & 1
		\end{array}
		\right)
		,\textnormal{ and }
		\left(\begin{array}{cc|cccc}
		1&0 & 2 & 0 & 0 & 0\\
		0&1 & 0 & 2 & 0 & 0\\
		0&0 & 3 & 3 & 1 & 0\\
		0&0 & 1 & 2 & 0 & 1
		\end{array}
		\right),
		$$
 respectively. We have that $\C_1\cap\C^\perp=\{\zero\}$ and $\C_1$ is an ACD code. Note that $2(0,1\mid2,3,1,0)*(1,1\mid 1,3,0,1)=(0,0\mid 0,2,0,0)\not \in \C_1$ and, therefore $C_1=\Phi(\C_1)$ is not linear and $C_1$ is not LCD. We have that neither $(\C_1)_\perp=\Phi(\C_1^\perp)$ is linear because, for example, $2(0,0\mid 3,3,1,0)*(0,0\mid 1,2,0,1)\not \in \C_1^\perp$. 
 \end{example}	
	
  \begin{example}\label{ex:ACD-LCD}
Now consider $\mathcal{C}_2$ the $\add$-additive code with generator and parity check matrix 
		$$
		\left(\begin{array}{cc|ccc}
		1&0 & 2 & 0 & 0\\
		0&1 & 2 & 2 & 0\\
		0 & 0 & 1 & 1 & 1
		\end{array}
		\right)
		,\textnormal{ and }
		\left(\begin{array}{cc|ccc}
		1&0 & 3 & 1 & 0\\
		1&1 & 3 & 0 & 1
		\end{array}
		\right),
		$$
 respectively. We have that $\C_2$ is an ACD code and $C_2=\Phi(\C_2)$ is linear. We have that the generator and parity check matrix of $C_2$ are
		$$
		\left(\begin{array}{cccccccc}
		1&0 & 0&0&1&1&1&1\\
		0&1 & 0&0&0&0&1&1\\
		0&0 & 1 &0& 1&0 & 1&0\\
		0&0&0&1&0&1&0&1
		\end{array}
		\right)
		,\textnormal{ and }
		\left(\begin{array}{cccccccc}
		1&0&0&1&0&1&0&0\\
		0&1&0&0&0&1&0&1\\
		0&0&1&1&0&0&1&1\\
		0&0&0&0&1&1&1&1
		\end{array}
		\right),
		$$
 respectively, $C_2\cap C_2^\perp=\{\zero\}$ and $C_2$ is LCD. However, since $2(1,0\mid 3,1,0)*(1,1\mid 3,0,1)=(0,0\mid 2,0,0)\not \in \C_2^\perp$, the code $(C_2)_\perp$ is not linear therefore it is not LCD. 
 \end{example}	

\begin{example}\label{ex:ACDnotLCD} Consider the $\add$-additive code $\C_3$ with generator and parity check matrix
		$$
		G=\left(\begin{array}{ccc|cccc}
		1&0&0&0&0&2&0\\
		0&1&0&0&0&2&2\\
		0&0&1&0&0&2&2\\
		0&0&0&1&1&0&1\\
		0&0&0&0&2&2&2
		\end{array}
		\right)
		,\textnormal{ and }
	    H=\left(\begin{array}{ccc|cccc}
	    0&0&0&2&2&0&0\\
	    1&1&1&3&1&1&0\\
	    0&1&1&2&1&0&1\\
	    \end{array}
		\right),
		$$
		respectively. The binary code $C_3=\Phi(\C_3)$ is linear, and it is easy to check that $(C_3)_\perp=\Phi(\C_3^\perp)$ is not linear. Then, we have that $(C_3)_\perp$ is not LCD. Moreover, $C_3\cap C_3^\perp=\langle(0,0,0,0,0,1,1,1,1,1,1)\rangle$ and, therefore, $C_3$ is not LCD.
\end{example}

The following property is easy to prove but the statement is new up to our knowledge.

\begin{lemma}\label{linear}
Let $\C\subset\Z_2^\alpha\times\Z_4^\beta$ be a $\add$-additive code, $C=\Phi(\C)$ and $C_\perp=\Phi(\C^\perp)$. If $C_\perp=C^\perp$, then $C$ is linear.
\end{lemma}
\begin{proof}
Clearly, $C\subseteq (C^\perp)^\perp=(C_\perp)^\perp$. Since $|C|\cdot|C_\perp|=2^{\alpha+2\beta}$, it follows that $C=(C_\perp)^\perp$ and $C$ is linear.
\end{proof}
	
\begin{lemma}\label{2*uvw} Let $\C$ be a $\add$-additive code. Let $\bv,\bw\in\C$ and $\bu\in\C^\perp$. Then, 
$$
[2\bv*\bw,\bu]=[2\bv*\bu,\bw]=[2\bu*\bw,\bv].
$$
\end{lemma}
\begin{proof}
Let  $\bv=(v\mid v'),\bw=(w\mid w')\in\C$ and $\bu=(u\mid u')\in\C^\perp$. We have that  $[2\bv*\bw,\bu]=0+2v'_1w'_1u'_1+\cdots+2v'_\beta w'_\beta u'_\beta$, that coincides with $[2\bv*\bu,\bw]$ and $[2\bu*\bw,\bv]$.
\end{proof}

\begin{proposition}\label{prop:2*CCperp}
Let $\C$ be a $\add$-additive code, $C=\Phi(\C)$. We have that $C$ is linear if and only if $2\bu*\bv\in\C^\perp$ for all $\bu\in\C,\bv\in\C^\perp$.
\end{proposition}
\begin{proof}
 We know that $C$ is linear if and only if $2 {\bf v} * \bw \in \C$, for all ${\bf v}, \bw\in\C$; that is, $[2 {\bf v} * \bw,\bu]=0$ for all $\bu\in\C^\perp$. By \Cref{2*uvw}, we have that $[2 {\bf v} * \bw,\bu]=[2\bv*\bu,\bw]$. Therefore, $C$ is linear if and only if  for all ${\bf v}, \bw\in\C$ and $\bu\in\C^\perp$, $[2\bv*\bu,\bw]=0$; that is equivalent to $2\bv*\bu\in \C^\perp$.
\end{proof}

\begin{corollary}\label{coro:2*CCperp}
    Let $\C$ be a $\add$-additive code, $C_\perp=\Phi(\C^\perp)$.  We have that $C_\perp$ is linear if and only if $2\bu*\bv\in\C$ for all $\bu\in\C,\bv\in\C^\perp$.
\end{corollary}
	
	Let $\C$ be a $\add$-additive code. Define the set
	$$
	D_\C=\{2\bu*\bv\mid \bu\in\C,\bv\in\C^\perp\}.
	$$
	
	With the definition of $D_\C$, we can obtain some corollaries of the previous proposition.
	\begin{corollary}\label{coro:Dc0}
	 Let $\C$ be a $\add$-additive, $C=\Phi(\C)$ and $C_\perp=\Phi(\C^\perp)$. If $D_\C=\{\zero\}$, then $C$ and $C_\perp$ are linear codes.
	\end{corollary}
	\begin{proof}
	If $D_\C=\{\zero\}$, then for all $\bu\in\C$ and $\bv\in\C^\perp$ we have that $2\bu*\bv=\zero$ and then $2\bu*\bv\in\C\cap\C^\perp$. By \Cref{prop:2*CCperp} and \Cref{coro:2*CCperp}, the codes $C$ and $C_\perp$ are linear.
	\end{proof}
	
	\begin{example} Let $\C_1$, $\C_2$, and $\C_3$ be the ACD codes defined in \Cref{ex:ACD-nonlinears}, \Cref{ex:ACD-LCD} and \Cref{ex:ACDnotLCD}, respectively. We have that  $\Phi(\C_1)$ and $\Phi(\C_1^\perp)$ are not linear. We have also seen that $\phi(\C_2)$ and $\Phi(\C_3)$ are linear whereas $\Phi(\C_2^\perp)$ and $\Phi(\C_3^\perp)$ are not linear. We obtain
	\begin{align*}
	D_{\C_1}&=\langle(0,0\mid 2,0,0,2),(0,0,\mid 0,2,0,2),(0,0\mid 0,0,2,2)\rangle,\\ 
	D_{\C_2}&=\langle(0,0\mid 2,0,2),(0,0\mid 0,2,2)\rangle, \textnormal{ and}\\
	D_{\C_3}&=\langle(0,0,0\mid 2,0,0,2),(0,0,0\mid 0,2,0,2)\rangle.
	\end{align*}
	Therefore, in none of this cases we obtain $D_\C=\{\zero\}$. 
	\end{example}
	
	\begin{corollary}\label{coro:Dc0sii}
	    Let $\C$ be an ACD code,  $C=\Phi(\C)$ and $C_\perp=\Phi(\C^\perp)$. Then, $C$ and $C_\perp$ are linear if and only if $D_\C=\{\zero\}$.
	\end{corollary}
	\begin{proof}
	    If $D_\C=\{\zero\}$, then $C$ and $C_\perp$ are linear codes by \Cref{coro:Dc0}.
	    
	    Assume now that $C$ and $C_\perp$ are linear codes. Let $2\bu*\bv\in D_\C$, where $\bu\in\C,\bv\in\C^\perp$. The code $C$ is linear and, by \Cref{prop:2*CCperp}, $2\bu*\bv\in\C^\perp$. Similarly, $C_\perp$ is linear and $2\bu*\bv\in\C$ by \Cref{coro:2*CCperp}. Then, $D_\C\subseteq\C\cap\C^\perp$. Since $\C$ is ACD, $\C\cap\C^\perp =\{\zero\}$ and $D_\C=\{\zero\}$.
	\end{proof}

		\begin{example}
		Let $\C$ be the ACD code given in \Cref{ex:G-H}. We have that $D_\C=\{\zero\}\in\C$ and hence C=$\Phi(\C)$ is linear by \Cref{coro:Dc0}. The code $C$ is a binary linear $(7,3)$ code with generator matrix and parity check matrix
		$$
		\left(\begin{array}{ccccccc}
		1 & 1 & 0 & 0 & 0 & 0 & 1 \\
		0 & 0 & 1 & 1 & 1 & 0 & 1 \\
		0 & 0 & 0 & 0 & 0 & 1 & 1 \\
		\end{array}
		\right), \textnormal{ and }
		\left(\begin{array}{ccccccc}
		1& 0 & 0 & 0 & 1 & 1 & 1 \\
		0& 1 & 0 & 0 & 1 & 1 & 1 \\
		0& 0 & 1 & 0 & 1 & 0 & 0 \\
		0& 0 & 0 & 1 & 1 & 0 & 0 \\
		\end{array}
		\right),
		$$
		respectively. Moreover, $C\cap C^\perp=\{\zero\}$ and therefore $C$ is an LCD code. In this case, we have that $C_\perp=C^\perp$ and hence $C_\perp$ is also LCD.\\
	\end{example}

	\begin{theorem}	\label{theo:unique}
		Let $\C\subseteq \Z_2^\alpha\times\Z_4^\beta$ be an ACD code. If $D_\C\subseteq\C\cup \C^\perp$, then any $x\in\Z_2^{\alpha+2\beta}$ can be written uniquely as $\Phi(\bu)+\Phi(\bv)$, for $\bu\in\C$, $\bv\in\C^\perp$.
	\end{theorem}
	\begin{proof}
	Let $x=\Phi(\bw)\in\Z_2^{\alpha+2\beta}$, for $\bw\in\Z_2^\alpha\times\Z_4^\beta$. Since $\C$ is ACD, there exist unique vectors $\bw_1\in\C$, $\bw_2\in\C^\perp$ such that $\bw=\bw_1+\bw_2$ by \Cref{sumadirecta}. Since $D_\C\subseteq\C\cup \C^\perp$,  we may assume $2\bw_1*\bw_2\in\C$; the case $2\bw_1*\bw_2\in\C^\perp$ is similar. By (\ref{equation:linear}), we have that $x=\Phi(\bw)=\Phi(\bw_1+\bw_2)=\Phi(\bw_1+2\bw_1*\bw_2)+\Phi(\bw_2)=\Phi(\bu)+\Phi(\bv)$, where $\bu=\bw_1+2\bw_1*\bw_2\in\C$ and $\bv=\bw_2\in\C^\perp$.
	
Note that $\bu$ and $\bv$ are computed in a unique way from the unique vectors $\bw_1$ and $\bw_2$, thus $\Phi(\bu)$ and $\Phi(\bv)$ are unique. Alternatively, if we assume that $\Phi(\bu)+\Phi(\bv)=\Phi(\bu')+\Phi(\bv')$, where $\bu,\bu'\in\C$, and $\bv,\bv'\in\C^\perp$, then we obtain
$$
\Phi(\bu+\bv+2\bu*\bv)=\Phi(\bu'+\bv'+2\bu'*\bv')\;\Longrightarrow\;\bu+\bv+2\bu*\bv=\bu'+\bv'+2\bu'*\bv'.
$$
Put $\bx=\bu-\bu'\in \C$ and $\by=\bv-\bv'\in\C^\perp$. Hence, we have
$$
\bx+\by=2\bu*\bv+2\bu'*\bv',
$$
where $\bx\in\C$, $\by\in\C^\perp$ and $2\bu*\bv,2\bu'*\bv'\in\C\cup\C^\perp$. Since $\C$ is an ACD code, we have that
\begin{itemize}
\item[(i)] $\bx = \zero$ ($2\bu*\bv,2\bu'*\bv'\in\C^\perp$), or
\item[(ii)] $\bx + 2\bu*\bv = \zero$ ($2\bu*\bv\in\C, 2\bu'*\bv'\in\C^\perp$), or
\item[(iii)] $\bx + 2\bu'*\bv' = \zero$ ($2\bu*\bv\in\C^\perp, 2\bu'*\bv'\in\C$), or
\item[(iv)] $\bx + 2\bu*\bv + 2\bu'*\bv' = \zero$ ($2\bu*\bv,2\bu'*\bv'\in\C$).
\end{itemize}

For case (i), we have $\bu=\bu'$ and therefore $\bv=\bv'$.

For case (ii), we have $\bu-\bu'+2\bu*\bv=\zero$ and $\bv-\bv'+2\bu'*\bv'=\zero$. Since $2\bu*\bv\in\C\cup\C^\perp$, we obtain that $\bu=\zero$ or $\bu'=\zero$ and thus $2\bu*\bv=\zero$ or $2\bu'*\bv'=\zero$. Consequently, $\bu=\bu'$ (implying $\bv=\bv'$) or $\bv=\bv'$ (implying $\bu=\bu'$).

Case (iii) is similar to (ii).

In case (iv), we have $\by=\zero$ and then it is similar to case (i).
\end{proof}

The following proposition gives a complete family of ACD codes whose images are LCD codes.

\begin{proposition}\label{Cristinaprop}
	Let $\C$ be the $\add$-additive code with generator matrix of the form	
	$$
	G=\left(G_X\mid I_\delta\right),
	$$
	where $G_X$ generates a self-orthogonal code. Then, $C=\Phi(\C)$ is a binary LCD code.
\end{proposition}	
\begin{proof}
	Let $\C$ be the $\add$-additive code with generator matrix $G=\left(G_X\mid I_\delta\right)$, such that $G_X$ generates a self-orthogonal code. Assume there exist $\bv=(v\mid v')\in\Z_2^\alpha\times\Z_4^\beta$ such that $\Phi(\bv)\in C\cap C^\perp$. Since $\Phi(\bv)\in C^\perp$, for all $\bw=(w\mid w')\in C$, $0=[\Phi(\bv),\Phi(\bw)]_2=[v,w]_2+[\Phi(v'),\Phi(w')]_2=0+[\Phi(v'),\Phi(w')]_2$, therefore $\Phi(v')=0$. Then, $\bv=(v\mid \zero)\in\C$ and, by the form of the generator matrix of $\C$, $v=\zero$ and $\bv=\zero$.
\end{proof}

\begin{corollary}
	For any non-negative integers $\alpha$ and $\delta\leq\lfloor \frac{\alpha}{2}\rfloor$, there exists an $(\alpha+2\delta,\delta)$ LCD code $C$ that is a $\add$-linear code of type $(\alpha,\delta;0,\delta;0)$. 
\end{corollary}

Note that in the previous proposition we have that the code $\Phi(\C)$ is linear and $D_\C\subseteq\C$. These are not the only ACD codes that have images that are LCD.

\begin{example}
	Consider the code generated by $(1 \ | \ 1).$
	This code has 4 vectors, namely  $(0 \ | \ 0), (1 \ | \ 1), (0 \ | \ 2)$ and $(1 \ | \ 3).$ Its orthogonal has two vectors, namely $(0 \ | \ 0)$ and $(1 \ | \ 2).$  Therefore, this code is ACD.  The images of these codes are
	$\{(0,0,0),(1,1,0),(0,1,1),(1,0,1)\}$
	and
	$\{(0,0,0),(1,1,1)\}$, which are binary LCD codes.  Therefore, we can have ACD codes whose images are LCD codes that do not satisfy the conditions of \Cref{Cristinaprop}.
\end{example}

Let $A$ and $B$ be two codes such that every element in the ambient space can be written uniquely as $a+b$ where $a \in A$ and $b \in B$.  For example, when $C$ is an ACD code then this is the case.  It is also the case, via Proposition~4.1 for binary codes that are the images of $\C$ and $\C^\perp$ when $\C$ is ACD and $D_\C \subseteq \C \cup \C^\perp.$ 
Define the complete bipartite graph $\Gamma = (V,E)$ where $V=A \cup B$.  Then the set of edges $E$ has cardinality $|A||B|$ which is the size of the ambient space.  Each edge represents the unique vector in the ambient space that is the sum of $a$ and $b$.  Therefore, this coding situation is equivalent to the complete bipartite graph $K_{|A|,|B|}.$ 

	\begin{lemma}\label{lemm:steven} Let $X, Y$ be binary vector spaces of length $n$ such that $X\cap X^\perp=\{0\}$ and $|X^\perp|=|Y|$. If any vector $a\in\Z_2^n$ have a unique representations $a=x+z$ and $a=x'+y$ where $x,x'\in X,z\in X^\perp$ and $y\in Y$, then $Y=X^\perp$. 
\end{lemma}
\begin{proof}
	Let $a\in\Z_2^n$ and consider the unique representations of $a$ of the form $a=x+z$ and $a=x'+y$ where $x,x'\in X,z\in X^\perp$ and $y\in Y$. We have that $x+z=x'+y$ and therefore $x+x'=z+y$. Then, since $x+x'\in X$, we have that $z+y\in X$. Moreover, $z+y\in\langle X^\perp,Y\rangle$, which is disjoint to $X$ except the element $\zero$. Hence $z+y=\zero$ that implies $y=z$. Therefore, for any $y\in Y$, we have $y=\zero+y$, $\zero\in X$, and $y\in X^\perp$. Since $|X^\perp|=|Y|$, $Y=X^\perp$.
\end{proof}

\begin{theorem}\label{theo:pep-commute}
	Let $\C$ be an ACD code and $C=\Phi(\C)$ a binary LCD code. Then, $C_\perp=\Phi(\C^\perp)$ is LCD and $C_\perp=C^\perp$.
\end{theorem}	

\begin{proof}
	Let $\C$ be an ACD such that $C$ is LCD. We have that $C\cap C^\perp=\{\zero\}$ and any $a\in\Z_2^{\alpha+2\beta}$ can be written uniquely as $a=x+z$, where $x\in C$ and $z\in C^\perp$. By \Cref{prop:2*CCperp}, $D_\C\subseteq \C^\perp$ and hence
	$a$ have a unique representation $a=x'+y$ where $x'\in C$ and $y\in C_\perp$ by \Cref{theo:unique}. Therefore, by \Cref{lemm:steven}, $C_\perp=C^\perp$.
\end{proof}

\begin{lemma}\label{2uv0}
	Let $\C$ be an ACD code such that $C=\Phi(\C)$ and $C_\perp=\Phi(\C^\perp)$ are linear. Then, for $\bu\in\C$ and $\bv\in\C^\perp$, we have that $\Phi(\bu+\bv)=\Phi(\bu)+\Phi(\bv)$.
\end{lemma}
\begin{proof}
Let $\bu\in\C$, $\bv\in\C^\perp$. We have that $\Phi(\bu+\bv)=\Phi(\bu)+\Phi(\bv)+\Phi(2\bu*\bv)$ and $2\bu*\bv\in D_\C$. Since $C$ and $C_\perp$ are linear, then $D_\C=\{\zero\}$ by \Cref{coro:Dc0sii}. Therefore,  $2\bu*\bv=\zero$ and $\Phi(\bu+\bv)=\Phi(\bu)+\Phi(\bv)$.
\end{proof}

\begin{lemma}\label{commute}
	Let $\bu,\bv\in\Z_2^\alpha\times\Z_4^\beta$ such that  $2\bu*\bv=\zero$. Then, $[\Phi(\bu),\Phi(\bv)]_2=0$ if and only if $[\bu,\bv]=0.$
\end{lemma}
\begin{proof}
	Let $\bu=(u\mid u'),\bv=(v\mid v')\in\Z_2^\alpha\times\Z_4^\beta$ such that  $2\bu*\bv=(\zero\mid 2u'*v')=\zero$. Consider $u'_i$ and $v'_i$ the $i$-th coordinate of $u'$ and $v'$ respectively. Note that if $2u'_iv'_i=0$, then $u'_i\in\{0,2\}$ or $v'_i\in\{0,2\}$. Let $J=\{1\leq j\leq \beta \mid u'_j\in\{1,3\} \textnormal{ or } v'_j\in\{1,3\}\}$. Note that if $i\in J$, then $u'_iv'_i=2$ and $[\Phi(u'_i),\Phi(v'_i)]_2=1$, and $u'_iv'_i=[\Phi(u'_i),\Phi(v'_i)]_2=0$, otherwise. Therefore, $[u',v']_4=\sum_{i=1}^\beta u'_iv'_i=\sum_{i\in J}u'_iv'_i=2\sum_{i\in J}[\Phi(u'_i),\Phi(v'_i)]_2=2[\Phi(u'),\Phi(v')]_2$, considering $[\Phi(u'),\Phi(v')]_2$ as an element in $\{0,1\}\subseteq \Z_4$.
		
	We have that $[\bu,\bv]=2[u,v]_2+[u',v']_4=2[u,v]_2+2[\Phi(u'),\Phi(v')]_2=2[\Phi(\bu),\Phi(\bv)]_2$, where $[\Phi(u'),\Phi(v')]_2$ and $[\Phi(\bu),\Phi(\bv)]_2$ are elements in $\{0,1\}\subseteq \Z_4$. Therefore, $[\Phi(\bu),\Phi(\bv)]_2=0$ if and only if $[\bu,\bv]=0.$
\end{proof}


\begin{theorem}\label{theo:conjecture}
Let $\C$ be an ACD code, $C=\Phi(\C)$ and $C_\perp=\Phi(\C^\perp)$. The following statements are equivalents:
\begin{enumerate}
    \item[(i)] $C$ is linear and $D_\C\subseteq \C$.
    \item[(ii)] $C_\perp$ is linear and $D_\C\subseteq \C^\perp$.
    \item[(iii)] $C$ and $C_\perp$ are linear.
    \item[(iv)] $D_\C=\{\zero\}$.
    \item[(v)] $C$ and $C_\perp$ are LCD.
    \item[(vi)] $C_\perp=C^\perp$.
\end{enumerate}
\end{theorem}
\begin{proof}
  By \Cref{coro:2*CCperp}, $C_\perp$ is linear if any only if $D_\C\subseteq \C$. Therefore we obtain $(i)\;\Leftrightarrow\;(iii)$. Similarly, by \Cref{prop:2*CCperp}, $C$ is linear if any only if $D_\C\subseteq \C^\perp$ and hence $(ii)\;\Leftrightarrow\;(iii)$. The equivalence $(iii)\;\Leftrightarrow\;(iv)$ is given by \Cref{coro:Dc0}.
  
  Now we will prove that $(iv)\;\Leftrightarrow\;(v)$. Let $\bv\in\C^\perp,\bu\in\C$. We have that $2\bu*\bv\in D_\C=\{\zero\}$. Then, by \Cref{commute}, $[\Phi(\bu),\Phi(\bv)]_2=[\bu,\bv]=0$ and, $\Phi(\bv)\in C^\perp$. Since $|\C^\perp|=|C^\perp|$, we have that $C_\perp=C^\perp$. Finally, if $x\in C\cap C^\perp$, then $\Phi^{-1}(x)\in \C\cap\C^\perp=\{\zero\}$, and therefore $x=\zero$. Then, since $C$ and $C_\perp$ are linear by (iii), both $C$ and $C_\perp$ are LCD.
  
  Finally, $(v)\;\Rightarrow\;(vi)$ by \Cref{theo:pep-commute}, and $(vi)\;\Rightarrow\;(iii)$ by \Cref{linear}.
\end{proof}

\begin{proposition}\label{free-LCD}
		Let $\C$ be a $\add$-additive code of type $(\alpha,\delta;0,\delta;0)$ generated by
	$$
	G=\left(G_X\mid I_\delta\right).
	$$
	Then, $C=\Phi(\C)$ is LCD.
\end{proposition}
\begin{proof}
	Let $\C$ be a $\add$-additive code generated by $G=\left(G_X\mid I_\delta\right)$, $C=\Phi(\C)$. It is easy to see that for all $\bu,\bv\in\Z_2^\alpha\times\Z_4^\beta$ we have that $2\bu* \bv\in\C$. Therefore, $C$ is linear and $D_\C\subseteq \C$. Then, by \Cref{theo:conjecture}, we have that $C$ is LCD.
\end{proof}	
\begin{corollary}
	Let $C$ be a binary $(\alpha,k)$ code and let $\{v_1,\dots,v_k\}$ be a basis for $C$. Let $\delta\geq k$ and let $G_X$ be the $\delta\times\alpha$ matrix whose non-zero row vectors are $\{v_1,\dots,v_k\}$. Then, the $\add$-additive code $\C$ of type $(\alpha,\delta;0,\delta;0)$ generated by
	$$
	G=\left(G_X\mid I_\delta\right).
	$$
	Then $\Phi(C)$ is LCD.
\end{corollary}

\begin{proof}
Straightforward by \Cref{theo:conjecture} and \Cref{theo:pep-commute}.
\end{proof}

\section{Summary and conclusions}\label{conclusions}

The concept and basic properties of binary LCD codes can be easily generalized to $\add$-additive complementary dual codes or ACD codes. 

Given an ACD code $\C$, we have considered the complementary duality of $\C_X$ and $\C_Y$. We have found examples of all possible situations:
\begin{itemize}
    \item Both $\C_X$ and $\C_Y$ are not LCD codes.
    \item Both $\C_X$ and $\C_Y$ are LCD codes. This is the only possibility when $\C$ is separable ($\C=\C_X \times \C_Y$), but if $\C$ is not separable, this situation is also possible.
    \item $\C_X$ is a LCD code and $\C_Y$ is not.
    \item $\C_Y$ is a LCD code and $\C_X$ is not.
\end{itemize}
The examples we have provided define, in several cases, infinite families of ACD codes in one of the previous situations.

We have studied the binary images of ACD codes obtained with the Gray map $\Phi$. For this, we have concluded that the set $D_\C=\{2\bu*\bv\mid u\in \C, v\in \C^\perp\}$ plays a very important role: $C=\Phi(\C)$ and $C_\perp=\Phi(\C^\perp)$ are LCD codes if and only if $D_\C=\{\zero\}$. 
	
\section*{Acknowledgements}
	
	This work has been partially supported by the Spanish grants TIN2016-77918-P, AEI/FEDER, UE.
	
	\nocite{*}
	
	\bibliographystyle{plain}

\end{document}